\documentclass[11pt, a4paper]{article}

\usepackage{amsmath, amssymb, amsthm, mathtools}
\usepackage{geometry}
\usepackage{hyperref}
\usepackage{xcolor}
\usepackage{booktabs}
\usepackage{newpxtext,newpxmath}

\theoremstyle{plain}
\newtheorem{theorem}{Theorem}[section]

\newtheorem{proposition}[theorem]{Proposition}

\theoremstyle{definition}
\newtheorem{definition}[theorem]{Definition}

\newtheorem{remark}[theorem]{Remark}

\newcommand{\CC}{\mathbb{C}}
\newcommand{\AOm}{\mathcal{A}_\Omega}
\newcommand{\Sh}{\mathrm{Sh}}

\renewcommand{\>}{\rangle}

\newcommand{\calP}{\mathcal{P}}

\title{\textbf{A Factorization Identity for\\
Twisted Multinomial Coefficients\\
with Application to Pilot States in\\
Hamiltonian Decoded Quantum Interferometry}}
\author{Pawe{\l} Wocjan\\[4pt]
\textit{IBM Quantum, T.J.\ Watson Research Center,
Yorktown Heights, NY 10598, USA}}
\date{April 15, 2026}

\begin{document}

\maketitle

\begin{abstract}
The $q$-multinomial coefficient, a classical object in enumerative
combinatorics, counts permutations of multisets weighted by the number
of inversions, with a single deformation parameter~$q$.  We introduce
the \emph{twisted multinomial coefficient}, in which each inversion
between letters $i$ and $j$ carries a pair-dependent weight
$\omega_{ij}$ determined by a skew-symmetric matrix~$\Omega$.
In general, no closed-form evaluation is known.  Our main result is that
under a natural structural condition on~$\Omega$ ---
\emph{predecessor-uniformity} ($\omega_{ij} = q_j$ for all $i<j$)
--- the twisted multinomial factorizes as a product of Gaussian
($q$-deformed) binomials with site-dependent parameters:
$\binom{k}{k_1,\ldots,k_m}_\Omega = \prod_j\binom{\ell_j}{k_j}_{q_j}$
where $\ell_j = k_1+\cdots+k_j$.
This extends the standard product formula for the $q$-multinomial
from a single parameter~$q$ to $m-1$ independent parameters.
The identity is purely combinatorial: it holds for arbitrary
$q_j \in \CC\setminus\{0\}$ without any algebraic constraints.

We were led to this identity by studying pilot state preparation in
Hamiltonian Decoded Quantum Interferometry (HDQI), a recently proposed
quantum algorithm for preparing Gibbs and ground states.  As an
application, we show that the factorization yields an exact matrix
product state (MPS) of bond dimension $k+1$ for the expansion
coefficients of $h^k$ in a twisted algebra.  We further show that the
same site matrices deliver an exact MPS of bond dimension
$\deg(\calP)+1$ for the expansion coefficients of $\calP(h)$, for any
polynomial~$\calP$, via a polynomial-dependent right boundary vector.

We note that this addresses only one component of the HDQI pipeline
(pilot state preparation); the full protocol additionally requires
efficient decoding of the associated Hamiltonian code, and both
components must work in conjunction for Hamiltonians of physical
interest.  Identifying such Hamiltonians remains an important open
problem.
\end{abstract}

\tableofcontents
\bigskip

\section{Introduction}

\subsection{The Combinatorial Problem}

The \emph{$q$-multinomial coefficient}
\begin{align}
  \binom{k}{k_1,\ldots,k_m}_q
  = \frac{[k]_q!}{[k_1]_q! \cdots [k_m]_q!},
  \qquad
  [s]_q = 1 + q + \cdots + q^{s-1},
  \quad
  [s]_q! = [1]_q [2]_q \cdots [s]_q,
\end{align}
is a classical object in enumerative combinatorics: it counts
permutations of the multiset $\{1^{k_1}, 2^{k_2}, \ldots, m^{k_m}\}$
weighted by $q$ raised to the number of
inversions~\cite[Section~1.3]{Stanley2011}
(see also~\cite{FoataHan} for a detailed treatment of inversion
statistics).
It factorizes as a product of Gaussian binomial coefficients,
\begin{align}
  \binom{k}{k_1,\ldots,k_m}_q
  = \prod_{j=1}^m \binom{\ell_j}{k_j}_q,
  \qquad \ell_j = k_1 + \cdots + k_j,
  \label{eq:standard_factorization}
\end{align}
each factor corresponding to the interleaving of letter $j$ into the
already-placed letters $1,\ldots,j{-}1$.  In this classical setting,
every inversion contributes the same weight~$q$, regardless of which pair
of letters is involved.
The Gaussian binomial $\binom{n}{r}_q$ is defined in full generality
in Definition~\ref{def:gaussian_binom}.

In this paper we introduce and study a generalization in which inversions
carry \emph{pair-dependent weights}.  Given a weight matrix
$\Omega = (\omega_{ij})$ with $\omega_{ii}=1$ and
$\omega_{ji} = \omega_{ij}^{-1}$ for $i < j$, the \emph{twisted
multinomial coefficient} is defined as
\begin{align}
  \binom{k}{k_1,\ldots,k_m}_\Omega
  :=
  \sum_{\sigma \in \mathrm{Sh}(k_1,\ldots,k_m)}
  \prod_{\substack{r < s \\ \sigma(r) > \sigma(s)}}
  \omega_{\sigma(s),\sigma(r)},
  \label{eq:twisted_intro}
\end{align}
where $\mathrm{Sh}(k_1,\ldots,k_m)$ denotes the set of all shuffles
(i.e.\ permutations of the multiset $\{1^{k_1},\ldots,m^{k_m}\}$),
the product runs over all inversions (pairs of positions $r < s$ with
$\sigma(r) > \sigma(s)$), and each inversion
contributes $\omega_{\sigma(s),\sigma(r)}$ with
$\sigma(s) < \sigma(r)$.
When $\omega_{ij} = q$ for all $i < j$, this reduces to the standard
$q$-multinomial.

In general, no clean factorization of~\eqref{eq:twisted_intro} is available:
the pair-dependence of the weights couples different stages of the
interleaving process.  Our main result
(Proposition~\ref{prop:factorization}) identifies a structural condition
on~$\Omega$ --- \emph{predecessor-uniformity} --- under which the
factorization~\eqref{eq:standard_factorization} extends to the
multi-parameter setting:
\begin{align}
  \binom{k}{k_1,\ldots,k_m}_\Omega
  = \prod_{j=1}^m \binom{\ell_j}{k_j}_{q_j},
  \label{eq:factorization_intro}
\end{align}
with a \emph{different} deformation parameter $q_j$ for each factor.
Predecessor-uniformity means that $\omega_{ij} = q_j$ for all
$i < j$: each letter~$j$ has a uniform weight with respect to all
earlier letters.  The identity~\eqref{eq:factorization_intro} is
purely combinatorial: it holds for arbitrary
$q_j \in \CC\setminus\{0\}$ and does not require any algebraic
relations.

To our knowledge, this multi-parameter factorization has not appeared
previously in the combinatorics literature; the standard
$q$-multinomial is the special case $q_j = q$ for all $j$.
We also give an $O(m^3)$ greedy algorithm that determines whether a
given weight matrix $\Omega$ admits a predecessor-uniform ordering
(Remark~\ref{rem:algorithm}).

\subsection{Motivation from Quantum Computing}

We were led to the twisted multinomial by studying a problem in quantum
computing: the preparation of \emph{pilot states} in Hamiltonian Decoded
Quantum Interferometry (HDQI), a recently proposed quantum algorithm for
preparing Gibbs states and ground states of many-body
Hamiltonians~\cite{Schmidhuber2025, Bu2026}.

Consider an $n$-qubit Pauli Hamiltonian
$H = \sum_{j=1}^m c_j P_j$ with $P_j \in \{I,X,Y,Z\}^{\otimes n}$.
The commutation phases $P_iP_j = \omega_{ij} P_jP_i$ define a twisting
matrix $\Omega = (\omega_{ij})$ with $\omega_{ij} \in \{+1,-1\}$.
The HDQI protocol requires the expansion of~$H^k$ in the generator basis:
\begin{align}
  H^k = \sum_{r \in \{0,1\}^m} \alpha_r\, P_1^{r_1}\cdots P_m^{r_m},
  \label{eq:expansion_intro}
\end{align}
and the coefficients $\alpha_r$ --- the \emph{pilot state amplitudes} ---
are precisely sums of twisted multinomial coefficients weighted by
products of the Hamiltonian's coefficients~$c_j$
(equation~\eqref{eq:alpha_twisted}).  It suffices to consider the
monomial $\calP(x) = x^k$: a general polynomial of degree at most~$k$
is handled by the \emph{same} site matrices at bond dimension $k+1$
with a polynomial-dependent right boundary vector
(Proposition~\ref{rem:direct_sum}).

The factorization~\eqref{eq:factorization_intro} implies that the
$\alpha_r$ admit an exact \emph{matrix product state} (MPS) of bond
dimension $k+1$ (Theorem~\ref{thm:twisted_mps}), enabling efficient
computation of the pilot state amplitudes.  In contrast, the existing
method of~\cite{Schmidhuber2025,Bu2026} decomposes the computation along
connected components of the \emph{anticommutation graph} (the graph on
$\{P_1,\ldots,P_m\}$ with an edge whenever two generators anticommute)
and has cost exponential in the size of the largest component.  We show
(Section~\ref{sec:comparison}) that predecessor-uniform twistings can
have an anticommutation graph that is a single connected component of
size~$m$, giving an exponential separation \emph{for the pilot state
computation}.

The HDQI protocol involves two further steps beyond pilot state
preparation: syndrome decoding and a coherent Bell
measurement~\cite{Schmidhuber2025,Bu2026}.  In this paper we focus
exclusively on the combinatorial and algebraic aspects of the pilot
state.  The question of which physically interesting Hamiltonians give
rise to predecessor-uniform twistings while also admitting efficient
decoding remains an important open problem.

\medskip
\textbf{Organisation.}
Section~\ref{sec:commutative} treats the commutative baseline.
Section~\ref{sec:twisted} is the heart of the paper: it defines the
twisted multinomial coefficient, introduces predecessor-uniformity,
proves the multi-parameter factorization into Gaussian binomials
(Proposition~\ref{prop:factorization}), derives the MPS
(Theorem~\ref{thm:twisted_mps}), and specializes to Pauli and qudit
settings.
Section~\ref{sec:comparison} compares with the connected-component
method and demonstrates an exponential separation in pilot state
computation cost for mutually anticommuting generators.
Section~\ref{sec:conclusions} discusses limitations and open problems.
Appendix~\ref{app:induction} gives an alternative induction proof of the
twisted multinomial theorem.

\section{The Commutative Setting}
\label{sec:commutative}

Let
\begin{align}
  \mathcal{A}
  =
  \CC[z_1,\dots,z_m]\big/\bigl(z_j^2-1,\; z_iz_j - z_jz_i\bigr).
\end{align}
The generators commute and square to $1$, so $\mathcal{A}$ has dimension $2^m$
with basis $\{z_1^{r_1}\cdots z_m^{r_m} : r_j \in \{0,1\}\}$.  Fix
$c_1,\dots,c_m\in\CC$ and set $h = \sum_{j=1}^m c_j z_j$.  We expand
$\calP(h) = h^k$ and seek an efficient representation of the amplitudes $\alpha_r$
appearing in $h^k = \sum_{r \in \{0,1\}^m} \alpha_r\,z_1^{r_1}\cdots z_m^{r_m}$.
This corresponds to the Pauli Hamiltonian case in which all generators $P_j$
pairwise commute, i.e.\ $\omega_{ij}=1$ for all $i\ne j$.  The coefficients
$\alpha_r$ are the pilot-state amplitudes studied in Section~3 of~\cite{Bu2026}
for commuting Hamiltonians; Theorem~\ref{thm:commutative} recovers their MPS
construction.

\begin{theorem}[Commutative MPS]
  \label{thm:commutative}
  The coefficient tensor $\alpha_r$ admits an exact MPS representation of bond
  dimension $k+1$ and physical dimension $2$.  Define matrices
  $A^{[j]}_{r_j} \in \CC^{(k+1)\times(k+1)}$ by
  \begin{align}
    \bigl(A^{[j]}_{r_j}\bigr)_{\ell_{j-1},\ell_j}
    =
    \begin{cases}
      \dfrac{c_j^{\ell_j-\ell_{j-1}}}{(\ell_j-\ell_{j-1})!}
      & \ell_j \ge \ell_{j-1}
        \text{ and }
        \ell_j - \ell_{j-1} \equiv r_j \pmod{2}, \\[6pt]
      0 & \text{otherwise,}
    \end{cases}
    \label{eq:comm_matrix}
  \end{align}
  and boundary vectors $v_0, v_m \in \CC^{k+1}$ by
  \begin{align}
    v_0 = (1, 0, \ldots, 0),
    \qquad
    v_m = (0, \ldots, 0, 1).
  \end{align}
  Then $\alpha_r = k! \cdot v_0 \bigl(\prod_{j=1}^m A^{[j]}_{r_j}\bigr)v_m^\top$.
\end{theorem}

\begin{proof}
By the multinomial theorem,
\begin{align}
  h^k
  =
  \sum_{\substack{k_1,\dots,k_m\ge 0\\k_1+\cdots+k_m=k}}
  \binom{k}{k_1,\dots,k_m}
  \prod_{j=1}^m c_j^{k_j} z_j^{k_j},
  \label{eq:comm_expansion}
\end{align}
where $\binom{k}{k_1,\ldots,k_m} = \frac{k!}{k_1!\cdots k_m!}$ is the
multinomial coefficient.
Since $z_j^2=1$, $z_j^{k_j} = z_j^{k_j\bmod 2}$, giving
\begin{align}
  \alpha_r 
  =
  \sum_{\substack{k_1+\cdots+k_m=k\\k_j\equiv r_j\pmod 2}}
  \binom{k}{k_1,\dots,k_m}
  \prod_{j=1}^m c_j^{k_j}.
\end{align}
Set $\ell_0=0$ and $\ell_j = k_1+\cdots+k_j$ for $j\in\{1,\ldots,m\}$.
The multinomial coefficient factorizes as
\begin{align}
  \binom{k}{k_1,\dots,k_m}
  = \frac{k!}{k_1!\cdots k_m!}
  = \prod_{j=1}^m \binom{\ell_j}{k_j},
  \label{eq:telescoping}
\end{align}
where the equality holds because
$\binom{\ell_j}{k_j} = \frac{\ell_j!}{k_j!\,\ell_{j-1}!}$ and the
product telescopes:
$\prod_j \frac{\ell_j!}{k_j!\,\ell_{j-1}!} = \frac{\ell_m!}{\ell_0!\prod_j k_j!} = \frac{k!}{\prod_j k_j!}$.
This factorization is the key to the MPS structure, and its twisted
generalization (Proposition~\ref{prop:factorization}) is the main result
of Section~\ref{sec:twisted}.
Substituting,
\begin{align}
  \alpha_r
  =
  k!\!
  \sum_{\substack{0=\ell_0\le\cdots\le\ell_m=k\\
      \ell_j-\ell_{j-1}\equiv r_j\pmod 2}}
  \prod_{j=1}^m \frac{c_j^{\ell_j-\ell_{j-1}}}{(\ell_j-\ell_{j-1})!}.
\end{align}
The summand is a product of nearest-neighbour factors in $\ell_j$,
giving the MPS with bond dimension $k+1$.
\end{proof}

A single $\alpha_r$ is computed by the sweep $v_j = v_{j-1} A^{[j]}_{r_j}$
at cost $O(mk^2)$.

\section{The Twisted Setting}
\label{sec:twisted}

\subsection{Twisted Algebra}

\begin{definition}[Single-order twisted algebra]
  \label{def:twisted_alg}
  Let $a \ge 2$ be an integer (the \emph{generator order}) and let
  $\Omega=(\omega_{ij})_{1\le i,j\le m}$ be a matrix of nonzero complex
  numbers satisfying
  \begin{align}
    \omega_{ii} = 1 \quad\text{for all } i,
    \qquad
    \omega_{ji} = \omega_{ij}^{-1} \quad\text{for } i < j,
    \qquad
    \omega_{ij}^a = 1 \quad\text{for all } i,j.
    \label{eq:omega_constraints}
  \end{align}
  The \emph{single-order twisted algebra} (the free associative algebra
  modulo the relations) is the quotient
  \begin{align}
    \AOm = \CC\<z_1,\dots,z_m\>\Big/\bigl(z_j^a-1,\;z_jz_i-\omega_{ij}z_iz_j\bigr).
    \label{eq:twisted_quotient}
  \end{align}
  Here $\CC\<z_1,\dots,z_m\>$ denotes the free associative algebra over $\CC$.
  Imposing the lexicographic ordering $z_1 < \cdots < z_m$, every element
  reduces to a unique linear combination of ordered monomials
  $\{z_1^{r_1}\cdots z_m^{r_m} : r_j\in\{0,1,\ldots,a-1\}\}$, giving an
  $a^m$-dimensional basis.

  The defining relation $z_j z_i = \omega_{ij}\,z_i z_j$ (for $j > i$)
  states that moving a later generator $z_j$ to the left of an earlier
  generator $z_i$ introduces the phase~$\omega_{ij}$.
\end{definition}

\begin{remark}[Consistency of the constraints]
  \label{rem:consistency}
  The third condition in~\eqref{eq:omega_constraints} is a consequence of
  the algebra.  From $z_j z_i = \omega_{ij}\,z_i z_j$ and $z_i^a = 1$,
  we get $z_j = z_j z_i^a = \omega_{ij}^a\, z_i^a z_j = \omega_{ij}^a\, z_j$,
  hence $\omega_{ij}^a = 1$: the commutation phases must be $a$-th
  roots of unity, i.e.\ $\omega_{ij} \in \mu_a := \{e^{2\pi ik/a}
  : k=0,\ldots,a-1\}$.
\end{remark}

\begin{remark}[Combinatorial versus algebraic constraints]
  \label{rem:generality}
  For the application to \emph{qubit} Pauli Hamiltonians, $a=2$ and the
  commutation phase between two $n$-qubit Paulis $P_i$ and $P_j$, encoded as
  $(u_j\mid w_j)\in\mathbb{F}_2^{2n}$ with $P = i^\kappa X^u Z^w$, is
  \begin{align}
    P_iP_j = (-1)^{\lambda_{ij}}P_jP_i,
    \qquad
    \lambda_{ij} = u_i\cdot w_j + w_i\cdot u_j \pmod{2},
    \label{eq:symplectic_comm}
  \end{align}
  so $\omega_{ij} = (-1)^{\lambda_{ij}} \in \mu_2 = \{+1,-1\}$.
  For \emph{qudit} Hamiltonians on systems of prime local dimension~$d$, the
  generalized Pauli (Heisenberg--Weyl) operators satisfy $a=d$ with
  $\omega_{ij} \in \mu_d$
  (see Section~\ref{subsec:qudit_specialization} for the role of primality).

  It is important to distinguish between the \emph{combinatorial} and
  \emph{algebraic} layers of the construction.  The twisted multinomial
  coefficient (Definition~\ref{def:twisted_multinomial}) and its
  factorization (Proposition~\ref{prop:factorization}) are purely
  combinatorial objects: they involve only the weights $\omega_{ij}$ and
  are well-defined for arbitrary
  $\omega_{ij}\in\CC\setminus\{0\}$ satisfying
  $\omega_{ii}=1$ and $\omega_{ji}=\omega_{ij}^{-1}$ for $i < j$,
  without requiring $\omega_{ij}^a=1$ or any algebraic relation
  $z_j^a=1$.  The algebraic constraint $\omega_{ij}^a = 1$ enters
  only when one \emph{interprets} the amplitudes~$\alpha_r$ as
  coefficients of $h^k$ in the algebra~$\AOm$, since the
  reduction $z_j^{k_j} = z_j^{k_j \bmod a}$ requires $z_j^a = 1$.
  We will return to this distinction in
  Remark~\ref{rem:combinatorial}.
\end{remark}

\begin{definition}[Shuffles and twisted multinomial coefficient]
  \label{def:twisted_multinomial}
  Let $k = k_1 + \cdots + k_m$.
  A $(k_1,\dots,k_m)$-\emph{shuffle} is a function
  $\sigma\colon\{1,\ldots,k\}\to\{1,\ldots,m\}$ with
  $|\sigma^{-1}(j)| = k_j$ for each $j$; equivalently, it is a word of
  length~$k$ in the alphabet $\{1,\dots,m\}$ with letter $j$ appearing
  exactly $k_j$ times.  Here the domain $\{1,\ldots,k\}$ indexes
  positions and the range $\{1,\ldots,m\}$ indexes letter types.
  We write $\Sh(k_1,\dots,k_m)$ for the set of all such shuffles.

  An \emph{inversion} is a pair of positions $(r,s)$ with $r<s$ and
  $\sigma(r)>\sigma(s)$; it carries the weight
  $\omega_{\sigma(s),\sigma(r)}$ (note $\sigma(s) < \sigma(r)$).
  This weight arises from the defining relation: reordering
  $z_{\sigma(r)} z_{\sigma(s)} = \omega_{\sigma(s),\sigma(r)}\,
  z_{\sigma(s)} z_{\sigma(r)}$ when $\sigma(r) > \sigma(s)$.
  The \emph{twisted multinomial coefficient} is
  \begin{align}
    \binom{k}{k_1,\dots,k_m}_\Omega
    :=
    \sum_{\sigma\in\Sh(k_1,\dots,k_m)}
    \prod_{\substack{r<s\\\sigma(r)>\sigma(s)}}
    \omega_{\sigma(s),\sigma(r)}.
    \label{eq:twisted_multinomial}
  \end{align}
\end{definition}

The following theorem shows that the twisted multinomial coefficient
plays the same role in the twisted algebra that the ordinary multinomial
coefficient plays in the commutative setting.  The expansion formula
appears implicitly in~\cite{Schmidhuber2025}; we state it here as a
standalone result for completeness and because the twisted multinomial
coefficient itself is the combinatorial object of independent interest.

\begin{theorem}[Twisted multinomial theorem]
  \label{thm:twisted_multinomial}
  Let $\AOm$ be a single-order twisted algebra and set
  $h = \sum_{j=1}^m c_j z_j$ with $c_j\in\CC$.  Then for every integer
  $k \ge 1$,
  \begin{align}
    h^k
    = \sum_{k_1+\cdots+k_m=k}
    \binom{k}{k_1,\dots,k_m}_\Omega \prod_{j=1}^m c_j^{k_j}\,
    z_1^{k_1}\cdots z_m^{k_m}.
    \label{eq:twisted_expansion}
  \end{align}
\end{theorem}

\begin{proof}
  Expand $h^k = \bigl(\sum_{j=1}^m c_j z_j\bigr)^k$ by distributing the
  product.  Each term in the expansion corresponds to a choice of one
  generator from each of the $k$ factors, yielding a word
  $z_{\sigma(1)} z_{\sigma(2)} \cdots z_{\sigma(k)}$ for some function
  $\sigma\colon\{1,\ldots,k\}\to\{1,\ldots,m\}$, with coefficient
  $c_{\sigma(1)} c_{\sigma(2)} \cdots c_{\sigma(k)}$.  We group these
  $m^k$ terms by the composition $(k_1,\ldots,k_m)$ with
  $k_j = |\sigma^{-1}(j)|$.  Within each group, the coefficient is
  $\prod_j c_j^{k_j}$ and the monomial is the disordered product
  $z_{\sigma(1)} \cdots z_{\sigma(k)}$, which we must rewrite in canonical
  form $z_1^{k_1}\cdots z_m^{k_m}$.

  To do so, apply pairwise transpositions using the relation
  $z_{\sigma(r)} z_{\sigma(s)} = \omega_{\sigma(s),\sigma(r)}\,
  z_{\sigma(s)} z_{\sigma(r)}$ whenever $\sigma(r) > \sigma(s)$ and
  $r < s$, i.e.\ whenever positions $(r,s)$ form an inversion.  Each
  such transposition contributes the phase $\omega_{\sigma(s),\sigma(r)}$.
  Since $\AOm$ is a well-defined quotient algebra (Definition~\ref{def:twisted_alg}),
  the defining relations are confluent: the canonical form and the total
  accumulated phase are independent of the order in which the transpositions
  are performed.  The total phase for the word~$\sigma$ is therefore
  \begin{align}
    \prod_{\substack{r<s\\\sigma(r)>\sigma(s)}}
    \omega_{\sigma(s),\sigma(r)},
  \end{align}
  which is precisely the inversion weight from
  Definition~\ref{def:twisted_multinomial}.  Summing over all words with
  composition $(k_1,\ldots,k_m)$ --- that is, over all shuffles
  $\sigma \in \Sh(k_1,\ldots,k_m)$ --- gives the twisted multinomial
  coefficient as the coefficient of
  $\prod_j c_j^{k_j}\,z_1^{k_1}\cdots z_m^{k_m}$.
  An alternative proof by induction on $k$ is given in
  Appendix~\ref{app:induction}.
\end{proof}

Reducing $z_j^{k_j}=z_j^{r_j}$ with $r_j=k_j\bmod a$:
\begin{align}
  \alpha_r
  = \sum_{\substack{k_1+\cdots+k_m=k\\k_j\equiv r_j\pmod a}}
  \binom{k}{k_1,\dots,k_m}_\Omega \prod_{j=1}^m c_j^{k_j},
  \qquad r \in \{0,1,\ldots,a-1\}^m.
  \label{eq:alpha_twisted}
\end{align}

\subsection{Predecessor-Uniform Twisting and Gaussian Binomial Factorization}

\begin{definition}[Predecessor-uniform twisting]
  \label{def:predecessor_uniform}
  A twisting matrix $\Omega$ satisfying~\eqref{eq:omega_constraints} is
  \emph{predecessor-uniform} if for each $j \in \{2,\ldots,m\}$ there exists
  $q_j \in \CC\setminus\{0\}$ such that $\omega_{ij}=q_j$ for all
  $i<j$.  (We set $q_1 = 1$ by convention, since the first generator has no
  predecessors.)
\end{definition}

\begin{remark}[Terminology and ordering dependence]
  \label{rem:ordering}
  The name ``predecessor-uniform'' reflects the defining property:
  each generator~$z_j$ has a uniform commutation parameter with all its
  predecessors $z_1,\ldots,z_{j-1}$ in the chosen ordering.  In the
  qubit case ($a=2$), $z_j$ either commutes with all predecessors
  ($q_j=+1$) or anticommutes with all of them ($q_j=-1$);
  for $a>2$, the generalized commutation phase $q_j$ must be the
  same $a$-th root of unity with respect to every predecessor.

  Predecessor-uniformity depends on the chosen ordering
  $z_1 < \cdots < z_m$ of the generators: a twisted algebra that is not
  predecessor-uniform under one ordering may become so under another.
  Given $\Omega$, one should seek a permutation $\pi\in S_m$ such that
  $\omega_{\pi(i),\pi(j)}$ depends only on $\pi(j)$ for all
  $\pi(i)<\pi(j)$.  Not all twisting matrices admit such a permutation
  (see Remark~\ref{rem:failure}).
\end{remark}

\begin{remark}[Efficient recognition of predecessor-uniform twistings]
  \label{rem:algorithm}
  Whether a twisting matrix $\Omega$ admits a predecessor-uniform ordering
  can be decided in $O(m^3)$ time by the following greedy peeling
  algorithm.  Build the permutation from the \emph{last} position
  backwards: at each step, find a generator $j$ in the remaining set~$S$
  whose commutation parameter $\omega_{ij}$ is constant over all
  $i\in S\setminus\{j\}$; place $j$ last among $S$ and remove it.
  If no such $j$ exists, no predecessor-uniform ordering is possible.

  \textit{Correctness.}  If a predecessor-uniform ordering $\pi$ of a
  set~$S$ exists, then restricting $\pi$ to any subset $S'\subseteq S$
  (preserving relative order) is again predecessor-uniform, because
  restriction only shrinks predecessor sets and a constant function
  remains constant on subsets.  In particular, the last element of~$\pi$
  is uniform with respect to all of $S\setminus\{\pi(|S|)\}$, so the
  greedy algorithm always finds a valid candidate.  After removing it,
  the restriction gives a valid ordering of the remainder.  By induction,
  the algorithm produces a predecessor-uniform ordering whenever one
  exists.

  \textit{Complexity.}  Each of the $m$ iterations considers at most
  $|S|$ candidates, each requiring $O(|S|)$ uniformity checks, giving
  $O(m^3)$ total.
\end{remark}

\begin{remark}[Failure for general $\Omega$]
  \label{rem:failure}
  Not all twisting matrices admit a predecessor-uniform ordering.
  For $m=3$ generators with $a=2$, \emph{every} twisting matrix
  admits a predecessor-uniform ordering: by the pigeonhole principle,
  among $\omega_{1j},\omega_{2j},\omega_{3j}\in\{+1,-1\}$ for any
  fixed $j$, at least two of the three generators share the same
  commutation sign with $j$, and a suitable ordering can always be found.

  \textit{Counterexample ($m=4$, $a=2$).}  Take $m=4$ generators with
  $\omega_{12}=-1$, $\omega_{13}=+1$, $\omega_{14}=-1$,
  $\omega_{23}=-1$, $\omega_{24}=+1$, $\omega_{34}=-1$.
  In this configuration, each generator commutes with exactly one of the
  other three and anticommutes with the remaining two.  For any
  generator~$j$, its commutation signs with the other three form a
  $\{+1,-1,-1\}$ or $\{-1,+1,+1\}$ pattern (never all equal), so no
  generator can be placed last in a predecessor-uniform ordering.  The
  greedy peeling algorithm (Remark~\ref{rem:algorithm}) immediately
  fails, and one can verify by exhaustive search that none of the $24$
  orderings satisfies Definition~\ref{def:predecessor_uniform}.

  A concrete realization on $n=2$ qubits is $P_1 = I\otimes X$,
  $P_2 = I\otimes Y$, $P_3 = X\otimes X$, $P_4 = X\otimes Y$;
  one verifies that the anticommutation graph is the $4$-cycle
  $P_1\text{--}P_2\text{--}P_3\text{--}P_4\text{--}P_1$ with
  commuting diagonals $(P_1,P_3)$ and $(P_2,P_4)$.  Two qubits are
  necessary: on a single qubit only three non-identity Paulis exist.
\end{remark}

The factorization of the twisted multinomial under predecessor-uniformity
involves Gaussian ($q$-deformed) binomial coefficients, which we now
define.

\begin{definition}[Gaussian binomial]
  \label{def:gaussian_binom}
  For $q \in \CC$ and integers $n \ge r \ge 0$, the \emph{Gaussian
  binomial coefficient} (or $q$-binomial coefficient) is defined as
  \begin{align}
    \binom{n}{r}_q
    := \sum_{\lambda \subseteq r \times (n-r)} q^{|\lambda|}.
    \label{eq:gauss_binom_partition}
  \end{align}
  Here the sum runs over all \emph{integer partitions} $\lambda$ that fit
  inside an $r \times (n-r)$ rectangular box.  A partition
  $\lambda = (\lambda_1 \ge \lambda_2 \ge \cdots \ge \lambda_r \ge 0)$
  fits in this box if it has at most $r$ parts, each of size at most $n-r$,
  i.e.\ $\lambda_1 \le n-r$.  Its \emph{size} is
  $|\lambda| = \lambda_1 + \cdots + \lambda_r$.

  \smallskip
  \textit{Examples.}
  $\binom{2}{1}_q = 1 + q$ (partitions fitting in a $1\times 1$ box:
  $\lambda=\varnothing$ with $|\lambda|=0$ and $\lambda=(1)$ with $|\lambda|=1$).
  $\binom{3}{2}_q = 1 + q + q^2$ (partitions fitting in a $2\times 1$ box:
  $\lambda = \varnothing, (1), (1{,}1)$ with sizes $0,1,2$).
  $\binom{4}{2}_q = 1 + q + 2q^2 + q^3 + q^4$ (partitions fitting in a
  $2\times 2$ box: $\lambda = \varnothing, (1), (2), (1{,}1), (2{,}1), (2{,}2)$
  with sizes $0,1,2,2,3,4$).

  \smallskip
  The expression~\eqref{eq:gauss_binom_partition} is manifestly a polynomial
  in $q$ with non-negative integer coefficients, and is therefore well-defined
  for all $q \in \CC$, including roots of unity.

  For $q$ not a root of unity, the Gaussian binomial admits the equivalent
  representations
  \begin{align}
    \binom{n}{r}_q
    = \frac{[n]_q!}{[r]_q!\,[n-r]_q!}
    = \prod_{i=0}^{r-1}\frac{1-q^{n-i}}{1-q^{i+1}},
    \label{eq:gauss_binom_product}
  \end{align}
  where $[s]_q = 1+q+\cdots+q^{s-1}$ is the $q$-integer and
  $[s]_q! = [1]_q [2]_q \cdots [s]_q$.  At roots of unity these product
  formulas can involve $0/0$ indeterminate forms, so we always take
  \eqref{eq:gauss_binom_partition} as the primary definition.

  The Gaussian binomial satisfies the $q$-Pascal recurrence
  \begin{align}
    \binom{n}{r}_q = \binom{n-1}{r-1}_q + q^r\binom{n-1}{r}_q,
    \label{eq:q_pascal}
  \end{align}
  with boundary values $\binom{n}{0}_q = 1$ for all $n \ge 0$
  and $\binom{n}{r}_q = 0$ for $r > n \ge 0$.
  These are the same conventions as for the ordinary binomial
  coefficients: $\binom{n}{0} = 1$ and $\binom{n}{r} = 0$ when $r > n$.
  In particular, $\binom{0}{0}_q = 1$ and $\binom{0}{r}_q = 0$ for
  $r \ge 1$.
  The recurrence~\eqref{eq:q_pascal} is well-defined for all $q$ and
  computes $\binom{n}{r}_q$ in $O(nr)$ operations, without the need for
  the product formula~\eqref{eq:gauss_binom_product}.
\end{definition}

We can now state the main factorization result.  In the
single-parameter case ($q_j = q$ for all $j$), the factorization
below reduces to the well-known product formula for the $q$-multinomial
coefficient; the combinatorial background can be found in Stanley's
treatment of permutations of multisets~\cite[Section~1.3, ``Permutations
of Multisets'', pp.~25--30]{Stanley2011}.  Our contribution is the
extension from a single parameter~$q$ to site-dependent
parameters~$q_j$, which is precisely what predecessor-uniformity
provides.

\begin{proposition}[Factorization of the twisted multinomial]
  \label{prop:factorization}
  Under predecessor-uniform twisting with parameters $q_j$,
  \begin{align}
    \binom{k}{k_1,\dots,k_m}_\Omega
    =
    \prod_{j=1}^m \binom{\ell_j}{k_j}_{q_j},
    \label{eq:factorization}
  \end{align}
  where $\ell_j=k_1+\cdots+k_j$ and $\binom{n}{r}_q$ is the Gaussian
  binomial (Definition~\ref{def:gaussian_binom}).
\end{proposition}

\begin{proof}
  We prove the identity by building each shuffle as a sequence of
  interleavings and showing that the total phase decomposes into
  independent contributions, one per stage.

  \textit{Step 1: Sequential construction of shuffles.}  Every
  $(k_1,\dots,k_m)$-shuffle can be constructed in $m$ stages.  At
  stage~$j = 1,\ldots,m$, insert $k_j$ copies of letter $j$ into the word
  of length $\ell_{j-1} = k_1+\cdots+k_{j-1}$ produced by the preceding
  stages.  After stage~$j$, the word has length $\ell_j = \ell_{j-1}+k_j$.
  Every shuffle arises exactly once from this procedure, and the inversions
  in the final word decompose into those created at each stage.

  \textit{Step 2: Phase per insertion (predecessor-uniformity).}  At
  stage~$j$, a single copy of $z_j$ is placed into the existing word.
  Suppose it lands with $t$ existing letters to its right.  Each of these
  $t$ letters carries some $z_i$ with $i < j$, creating an inversion at
  position pair $(z_j, z_i)$.  The associated phase is
  $\omega_{ij}$ (the weight for the pair $i < j$).  By
  predecessor-uniformity, $\omega_{ij} = q_j$ for every $i<j$, so
  the phase is $q_j$ regardless of which predecessor $z_i$ is involved.
  The $t$ inversions therefore contribute a combined phase of $q_j^t$.

  This is the crucial step: without predecessor-uniformity, the phase would
  depend on the identities $z_{i_1},\ldots,z_{i_t}$ of the letters to the
  right, not just their count~$t$, and the factorization would fail (cf.\
  Remark~\ref{rem:failure}).

  \textit{Step 3: Interleaving sum at stage $j$.}  At stage $j$, we
  insert $k_j$ copies of $z_j$ into a word of length $\ell_{j-1}$.  The
  $u$-th copy ($u=1,\ldots,k_j$) lands with $t_u$ existing letters to its
  right.  Since the copies of $z_j$ are identical, we may assume they are
  inserted from right to left without loss of generality, giving a weakly
  decreasing sequence
  $\ell_{j-1} \ge t_1 \ge t_2 \ge \cdots \ge t_{k_j} \ge 0$.
  Equivalently, the $k_j$ insertion positions, read in non-increasing order,
  form a partition $\lambda = (t_1,\ldots,t_{k_j})$ fitting inside a
  $k_j \times \ell_{j-1}$ box.

  From Step~2, the total phase contributed at stage~$j$ is
  $q_j^{t_1+t_2+\cdots+t_{k_j}} = q_j^{|\lambda|}$.
  Summing over all valid partitions:
  \begin{align}
    \sum_{\lambda \subseteq k_j \times \ell_{j-1}}
    q_j^{\,|\lambda|}
    = \binom{\ell_{j-1}+k_j}{k_j}_{q_j}
    = \binom{\ell_j}{k_j}_{q_j},
    \label{eq:interleaving_sum}
  \end{align}
  which is precisely the Gaussian binomial
  (Definition~\ref{def:gaussian_binom}), since the sum runs over
  partitions with at most $k_j$ parts each of size at most~$\ell_{j-1}$,
  weighted by $q_j^{|\lambda|}$.

  \textit{Step 4: Independence across stages.}  The total phase of a
  shuffle is the product of phases from each stage.  Crucially, the
  interleaving sum at stage $j$ (equation~\eqref{eq:interleaving_sum})
  depends on the preceding stages only through the total count
  $\ell_{j-1} = k_1+\cdots+k_{j-1}$, not on the identities of the
  $\ell_{j-1}$ letters already placed.  This is a direct consequence of
  predecessor-uniformity (Step~2): all predecessors contribute the same
  phase $q_j$, so only their count matters.  Therefore the contributions at
  different stages are independent and multiply:
  \begin{align}
    \binom{k}{k_1,\dots,k_m}_\Omega
    = \prod_{j=1}^m \binom{\ell_j}{k_j}_{q_j}.
    \tag*{\qedhere}
  \end{align}
\end{proof}

\begin{remark}[Commutative limit]
  \label{rem:commutative_limit}
  At $q_j=1$ for all $j$: $\binom{\ell_j}{k_j}_1 = \binom{\ell_j}{k_j}$, and
  $\prod_j\binom{\ell_j}{k_j}$ telescopes to $k!/(k_1!\cdots k_m!)$
  as in \eqref{eq:telescoping},
  recovering the commutative factorization.
\end{remark}

\begin{remark}[Generality of the factorization]
  \label{rem:combinatorial}
  As noted in Remark~\ref{rem:generality},
  Proposition~\ref{prop:factorization} belongs to the combinatorial
  layer: it holds for all $q_j \in \CC\setminus\{0\}$ with no
  requirement that $q_j^a = 1$ or that any algebra $\AOm$ exists.
  The algebraic constraint $q_j \in \mu_a$ is needed only when one
  wishes to interpret the resulting amplitudes~$\alpha_r$ as expansion
  coefficients in a twisted algebra with $z_j^a = 1$.
\end{remark}

\subsection{MPS Representation}

\begin{theorem}[MPS for single-order predecessor-uniform twisted algebras]
  \label{thm:twisted_mps}
  Let $\AOm$ be a single-order twisted algebra with generator order $a \ge 2$
  and predecessor-uniform twisting with parameters $q_j$.
  The physical indices take values $r_j \in \{0,1,\ldots,a-1\}$.  Define
  $\widetilde{A}^{[j]}_{r_j}\in\CC^{(k+1)\times(k+1)}$ by
  \begin{align}
    \bigl(\widetilde{A}^{[j]}_{r_j}\bigr)_{\ell_{j-1},\ell_j}
    =
    \begin{cases}
      c_j^{\,\ell_j-\ell_{j-1}}
      \dbinom{\ell_j}{\ell_j-\ell_{j-1}}_{q_j}
      & \ell_j\ge\ell_{j-1},\;
        \ell_j-\ell_{j-1}\equiv r_j\pmod{a}, \\[8pt]
      0 & \text{otherwise.}
    \end{cases}
    \label{eq:twisted_matrix}
  \end{align}
  Then
  \begin{align}
    \alpha_r
    = 
    v_0
    \Bigl(\prod_{j=1}^m \widetilde{A}^{[j]}_{r_j}\Bigr)
    v_m^\top,
    \label{eq:mps_twisted}
  \end{align}
  where the boundary vectors $v_0 = (1,0,\ldots,0)$ and
  $v_m = (0,\ldots,0,1)$ are as in the commutative case.
\end{theorem}

\begin{proof}
  Substituting \eqref{eq:factorization} into \eqref{eq:alpha_twisted} and
  re-indexing by $\ell_j$:
  \begin{align}
    \alpha_r
    =
    \sum_{\substack{0=\ell_0\le\cdots\le\ell_m=k\\
        \ell_j-\ell_{j-1}\equiv r_j\pmod a}}
    \prod_{j=1}^m
    c_j^{\ell_j-\ell_{j-1}}
    \binom{\ell_j}{\ell_j-\ell_{j-1}}_{q_j}.
  \end{align}
  Each summand is a product of nearest-neighbour factors; encoding endpoint
  constraints via $v_0$, $v_m$ yields \eqref{eq:mps_twisted}
  with bond dimension $k+1$.
\end{proof}

\begin{remark}[Sparsity and computational cost]
  \label{rem:cost}
  A single amplitude $\alpha_r$ is computed by the sweep
  $v_j = v_{j-1}\widetilde{A}^{[j]}_{r_j}$ at cost $O(mk^2)$.  For
  general $q_j$, the Gaussian binomial entries can be precomputed via the
  $q$-Pascal recurrence~\eqref{eq:q_pascal} in $O(k^2)$ per site and
  $O(mk^2)$ total.

  The selection rule $k_j \equiv r_j \pmod{a}$ means that each matrix
  $\widetilde{A}^{[j]}_{r_j}$ has approximately a fraction $1/a$ of its
  entries nonzero; in particular, larger generator orders produce sparser
  MPS matrices.  In practice, one does not need to enumerate all $a^m$
  amplitudes: the MPS representation~\eqref{eq:mps_twisted} with bond
  dimension $k+1$ can be used directly for efficient pilot state
  preparation, since known methods can prepare a quantum state described
  by an MPS of bond dimension $D$ in $\mathrm{poly}(m,D)$ gates
  (see, e.g.,~\cite{Schollwoeck2011}).
\end{remark}

\begin{proposition}[Extension to general polynomials]
  \label{rem:direct_sum}
  Let $\calP(x) = \sum_{j=0}^d a_j x^j$ be an arbitrary polynomial of
  degree at most~$d$, and let $\widetilde{A}^{[j']}_{r_{j'}}$ denote the
  site matrices of Theorem~\ref{thm:twisted_mps} built at bond
  dimension $d+1$ (i.e.\ with~$k$ in \eqref{eq:twisted_matrix} replaced
  by~$d$).  Then the pilot state amplitudes
  $\alpha_r(\calP) = \sum_{j=0}^d a_j\,\alpha_r^{(j)}$ admit an exact
  MPS representation of bond dimension $d+1$ using these same site
  matrices, the same left boundary $v_0 = (1,0,\ldots,0)\in\CC^{d+1}$,
  and the $\calP$-dependent right boundary vector
  \begin{align}
    v_m^{(\calP)} \;=\; (a_0,\,a_1,\,\ldots,\,a_d)\;\in\;\CC^{d+1}.
  \end{align}
  Explicitly,
  \begin{align}
    \alpha_r(\calP)
    \;=\;
    v_0\,\Bigl(\prod_{j'=1}^m \widetilde{A}^{[j']}_{r_{j'}}\Bigr)\,
    \bigl(v_m^{(\calP)}\bigr)^{\!\top}.
    \label{eq:poly_mps}
  \end{align}
  This improves on the standard direct-sum construction for linear
  combinations of MPS (block-diagonal site matrices with appropriate
  boundary vectors; see, e.g.,
  Schollw\"ock~\cite{Schollwoeck2011}), which applied to the sum of
  monomial pilot states of degrees $0,1,\ldots,d$ would yield bond
  dimension $\sum_{j=0}^d(j{+}1)=\binom{d+2}{2}$.
\end{proposition}

\begin{proof}
  The matrix entry
  $\bigl(\widetilde{A}^{[j']}_{r_{j'}}\bigr)_{\ell_{j'-1},\ell_{j'}}$
  in \eqref{eq:twisted_matrix} is
  $c_{j'}^{\ell_{j'}-\ell_{j'-1}}\,
   \binom{\ell_{j'}}{\ell_{j'}-\ell_{j'-1}}_{q_{j'}}$ whenever
  $\ell_{j'}\ge\ell_{j'-1}$ and $\ell_{j'}-\ell_{j'-1}\equiv r_{j'}\pmod a$,
  and zero otherwise.  In particular, it depends only on the bond
  indices $\ell_{j'-1},\ell_{j'}\in\{0,1,\ldots,d\}$ and on the physical
  index $r_{j'}$; it is \emph{independent} of the overall polynomial
  degree chosen in Theorem~\ref{thm:twisted_mps}.

  For each $j_0 \in \{0,1,\ldots,d\}$, apply
  Theorem~\ref{thm:twisted_mps} with degree $k = j_0$, embedded in
  $\CC^{d+1}$ by padding with zeros at bond levels $\ell > j_0$.  Since
  the bond indices can only increase through the MPS contraction and
  the left boundary is $v_0=e_0$, any contraction that reaches bond
  level $\ell_m > j_0$ does not contribute to the $j_0$-th monomial
  amplitude, while the contraction that ends at $\ell_m = j_0$ computes
  exactly $\alpha_r^{(j_0)}$.  Hence
  \begin{align}
    \alpha_r^{(j_0)}
    \;=\;
    v_0\,\Bigl(\prod_{j'=1}^m \widetilde{A}^{[j']}_{r_{j'}}\Bigr)\,
    e_{j_0}^{\!\top},
  \end{align}
  where $e_{j_0}\in\CC^{d+1}$ is the $(j_0{+}1)$-th standard basis
  vector.  By linearity of $\alpha_r$ in the polynomial~$\calP$,
  \begin{align}
    \alpha_r(\calP)
    \;=\;
    \sum_{j_0=0}^d a_{j_0}\,\alpha_r^{(j_0)}
    \;=\;
    v_0\,\Bigl(\prod_{j'=1}^m \widetilde{A}^{[j']}_{r_{j'}}\Bigr)\,
    \Bigl(\sum_{j_0=0}^d a_{j_0}\,e_{j_0}\Bigr)^{\!\top}
    \;=\;
    v_0\,\Bigl(\prod_{j'=1}^m \widetilde{A}^{[j']}_{r_{j'}}\Bigr)\,
    \bigl(v_m^{(\calP)}\bigr)^{\!\top}.
  \end{align}
  The bond-dimension count of the site matrices is $(d{+}1)\times(d{+}1)$
  by construction.
\end{proof}

\subsection{Specialization to Pauli Hamiltonians ($a=2$)}
\label{subsec:pauli_specialization}

For Pauli Hamiltonians on qubits, $a=2$ and
$\omega_{ij} = (-1)^{\lambda_{ij}} \in \{+1,-1\}$.
The predecessor-uniform parameters satisfy $q_j \in \{+1,-1\}$.
The physical indices are $r_j \in \{0,1\}$, giving a $2^m$-dimensional
pilot state.  Only two Gaussian binomials arise.

\textit{Case $q_j = +1$ (commuting generator).}  This reduces to the
ordinary binomial coefficient: $\binom{n}{r}_1 = \binom{n}{r}$.

\textit{Case $q_j = -1$ (anticommuting generator).}  The Gaussian binomial
at $q=-1$ admits the well-known closed form
\begin{align}
  \binom{n}{r}_{-1}
  =
  \begin{cases}
    \displaystyle\binom{\lfloor n/2\rfloor}{\lfloor r/2\rfloor}
      & \text{if } n \text{ is odd or } r \text{ is even}, \\[6pt]
    0 & \text{if } n \text{ is even and } r \text{ is odd}.
  \end{cases}
  \label{eq:gauss_binom_minus1}
\end{align}
The vanishing pattern --- $\binom{n}{r}_{-1} = 0$ whenever $n$ is even and
$r$ is odd --- introduces sparsity in the MPS matrices
$\widetilde{A}^{[j]}_{r_j}$ for anticommuting generators, which can be
exploited to reduce the cost of the MPS contraction.

Both cases allow $O(1)$ evaluation per entry (using precomputed ordinary
binomials), so all MPS matrix entries can be filled in $O(mk^2)$ total
time --- the same asymptotic cost as a single MPS sweep.

\begin{remark}[Predecessor-uniformity for Paulis]
  \label{rem:pauli_predecessor_uniform}
  In the Pauli case ($a=2$), predecessor-uniformity means that for each
  $j$, the generator $P_j$ either commutes with all of
  $P_1,\ldots,P_{j-1}$ ($q_j = +1$) or anticommutes with all of
  them ($q_j = -1$).  For example, a set of mutually anticommuting
  Paulis is predecessor-uniform under any ordering (with $q_j = -1$
  for all $j \ge 2$).
\end{remark}

\subsection{Specialization to Qudit Hamiltonians ($a=d$, $d$ prime)}
\label{subsec:qudit_specialization}

For qudit systems of prime local dimension~$d$, the generalized Pauli
(Heisenberg--Weyl) operators satisfy $Z_j^d = 1$, with commutation phases
$\omega_{ij} \in \mu_d = \{e^{2\pi ik/d}: k=0,\ldots,d-1\}$.  The
primality of~$d$ is essential for the single-order framework: the clock
and shift operators $X$ and $Z$ satisfy $X^d = Z^d = I$, so every
generalized Pauli $P = \omega^c X^a Z^b$ has order dividing~$d$.  When
$d$ is prime, the only divisors are $1$ and $d$, so every non-identity
generalized Pauli has order exactly~$d$ and the single-order condition
$a = d$ holds uniformly.  For composite~$d$, different generators may
have different orders (e.g.\ for $d=6$, the operator $X^3$ has order~$2$),
and the single-order assumption would not apply without further
restrictions on the generator set.

The generator order is thus $a = d$, the physical indices take values
$r_j \in \{0,1,\ldots,d-1\}$, and the pilot state has $d^m$ amplitudes.

Unlike the Pauli case ($d=2$), the Gaussian binomials at $q = e^{2\pi ik/d}$
do not generally admit simple closed forms.  However, the $q$-Pascal
recurrence~\eqref{eq:q_pascal} computes each $\binom{n}{r}_q$ in $O(nr)$
operations regardless of the value of~$q$, including at roots of unity.
The MPS bond dimension remains $k+1$ and the contraction cost is still
$O(mk^2)$ per amplitude.

\section{Comparison with the Connected-Component Approach}
\label{sec:comparison}

We briefly compare the MPS construction of Theorem~\ref{thm:twisted_mps}
with the approach based on connected components of the anticommutation
graph, which is the method used in~\cite{Schmidhuber2025,Bu2026}.

\begin{definition}[Anticommutation graph]
  \label{def:anticomm_graph}
  The \emph{anticommutation graph} $G=(V,E)$ of a set of Pauli generators
  $\{P_1,\ldots,P_m\}$ has vertex set $V=\{1,\ldots,m\}$ and an edge
  $(i,j)\in E$ whenever $P_iP_j = -P_jP_i$.
\end{definition}

Suppose $G$ decomposes into connected components
$C_1,\ldots,C_p$.  In the approach of~\cite{Schmidhuber2025,Bu2026},
the computation of the pilot state amplitudes factorizes over
components, and the cost is dominated by the largest component.  Let
$c_{\max} = \max_\ell |C_\ell|$ denote the size of the largest connected
component.  The overall cost of the connected-component method
scales exponentially in $c_{\max}$.

By contrast, the MPS of Theorem~\ref{thm:twisted_mps} has bond dimension
$k+1$ and computes each amplitude in $O(mk^2)$ time, \emph{irrespective}
of the structure of the anticommutation graph, provided the twisting
matrix admits a predecessor-uniform ordering.  The following example shows
that predecessor-uniform twistings can have $c_{\max} = m$, giving an
exponential separation in the cost of computing pilot state amplitudes.

\begin{proposition}[Mutually anticommuting generators]
  \label{prop:anticommuting}
  Let $\{P_1,\ldots,P_m\}$ be a set of mutually anticommuting $n$-qubit
  Pauli operators (such a set exists with $m = 2n+1$).  Then:
  \begin{enumerate}
    \item[\emph{(i)}] The twisting matrix $\Omega$ with
      $\omega_{ij}=-1$ for all $i \ne j$ is predecessor-uniform under
      every ordering (with $q_j=-1$ for all $j\ge 2$).
    \item[\emph{(ii)}] The anticommutation graph is the complete graph
      $K_m$, which is a single connected component with $c_{\max} = m$.
    \item[\emph{(iii)}] The connected-component approach
      of\/~\cite{Schmidhuber2025,Bu2026} has cost exponential in~$m$,
      while the MPS of Theorem~\ref{thm:twisted_mps} has bond dimension
      $k+1$ and cost $O(mk^2)$ per amplitude.
  \end{enumerate}
\end{proposition}

\begin{proof}
  Part~\emph{(i)}: since $\omega_{ij}=-1$ for all $i < j$, for any
  ordering and any generator $z_j$, the commutation parameter with every
  predecessor is $-1$; this is constant, so predecessor-uniformity holds
  trivially.  Part~\emph{(ii)} is immediate: every pair of generators is
  connected by an edge.  Part~\emph{(iii)} follows from
  Theorem~\ref{thm:twisted_mps} with $q_j = -1$ for all $j\ge 2$.
\end{proof}

An explicit family of mutually anticommuting $n$-qubit Paulis of size
$m = 2n+1$ is given by the Jordan--Wigner construction: for
$k = 1,\ldots,n$, define
\begin{align}
  \gamma_{2k-1} = Z^{\otimes(k-1)}\otimes X \otimes I^{\otimes(n-k)},
  \qquad
  \gamma_{2k} = Z^{\otimes(k-1)}\otimes Y \otimes I^{\otimes(n-k)},
\end{align}
together with $\gamma_{2n+1} = Z^{\otimes n}$.  One verifies that
$\gamma_i\gamma_j = -\gamma_j\gamma_i$ for all $i\ne j$.  For example,
on $n=2$ qubits the five generators are $XI$, $YI$, $ZX$, $ZY$, $ZZ$;
on $n=3$ qubits there are seven: $XII$, $YII$, $ZXI$, $ZYI$, $ZZX$,
$ZZY$, $ZZZ$.

\begin{remark}[Compatibility with the symplectic structure]
  \label{rem:symplectic}
  One might wonder whether the symplectic origin of the Pauli commutation
  phases imposes additional constraints that make predecessor-uniform
  twistings harder to realize than the abstract algebra suggests.

  Predecessor-uniformity is a condition on the entries of the
  $m\times m$ matrix $\Lambda = (\lambda_{ij})$, where
  $\lambda_{ij} = u_i\cdot w_j + w_i\cdot u_j\pmod{2}$ is the
  symplectic form of $P_i = (u_i\mid w_i)$ and
  $P_j = (u_j\mid w_j)$ in $\mathbb{F}_2^{2n}$.  The symplectic form
  automatically ensures $\Lambda_{ii}=0$ and $\Lambda_{ij}=\Lambda_{ji}$,
  and one might ask whether further constraints are present --- for
  instance, whether the values $\lambda_{12}$, $\lambda_{23}$ could
  force $\lambda_{13}$.

  No such constraint exists: the $\binom{m}{2}$ upper-triangular entries
  of $\Lambda$ are freely specifiable.  Given any $m\times m$ symmetric
  $\mathbb{F}_2$ matrix with zero diagonal, the following
  construction on $n=m$ qubits realizes it.  Set
  $P_i = (u_i\mid w_i)\in\mathbb{F}_2^{2m}$ with $u_i = e_i$
  (the $i$-th standard basis vector) and
  \begin{align}
    w_i[j] =
    \begin{cases}
      \Lambda_{ij} & \text{if } j > i, \\
      0            & \text{if } j \le i.
    \end{cases}
    \label{eq:realization}
  \end{align}
  For $i<j$, the symplectic product is
  $\langle P_i, P_j \rangle_s
   = \underbrace{e_i \cdot w_j}_{=\,w_j[i]\,=\,0}
   + \underbrace{w_i \cdot e_j}_{=\,w_i[j]\,=\,\Lambda_{ij}}
   = \Lambda_{ij}$,
  where $w_j[i]=0$ because $i < j$ and $w_j$ is supported only on
  indices $>j$.

  In particular, the all-ones off-diagonal matrix
  ($\Lambda_{ij}=1$ for $i\ne j$) that encodes mutual anticommutation
  is realized by the Jordan--Wigner construction on
  $n = \lceil(m-1)/2\rceil$ qubits, confirming that
  Proposition~\ref{prop:anticommuting} applies to genuine Pauli
  Hamiltonians.
\end{remark}

\begin{remark}[Predecessor-uniformity and locality]
  \label{rem:locality}
  The most significant limitation of predecessor-uniformity from a
  physical perspective is its tension with locality.
  In the qubit case ($a=2$), if generator $P_j$ anticommutes
  with any predecessor $P_i$ ($i < j$), then it must anticommute with
  \emph{all} predecessors $P_1,\ldots,P_{j-1}$.  In particular, $P_j$
  must share support with every preceding generator.  This rules out
  geometrically local Hamiltonians, where each generator acts on a bounded
  number of qubits and therefore commutes with all generators whose
  support is disjoint.

  More precisely, for a Hamiltonian on a lattice with generators of
  bounded support, most pairs of generators have disjoint support and thus
  commute.  For any generator $P_j$ that anticommutes with even one
  predecessor, predecessor-uniformity would force $P_j$ to anticommute
  with all earlier generators --- including those with disjoint support,
  which is impossible.  The predecessor-uniform condition therefore selects
  Hamiltonians with highly non-local, overlapping generators.

  This does not render the method useless: mutually anticommuting
  generators (Proposition~\ref{prop:anticommuting}) are a natural class
  that arises, for instance, in the Jordan--Wigner encoding of fermionic
  systems.  However, the tension with locality highlights that the
  greatest potential for the MPS approach lies in identifying Hamiltonians
  that combine non-local generator structure with tractable decoding.
\end{remark}

\section{Conclusions}
\label{sec:conclusions}

We have introduced the twisted multinomial coefficient --- a
generalization of the classical $q$-multinomial in which inversions
carry pair-dependent weights --- and proved that under
predecessor-uniformity, it factorizes as a product of Gaussian binomials
with site-dependent parameters $q_j$
(Proposition~\ref{prop:factorization}).  This purely combinatorial
identity yields an exact MPS of bond dimension $k+1$ for the pilot
state amplitudes of $h^k$ in the HDQI framework
(Theorem~\ref{thm:twisted_mps}), and more generally bond dimension
$\deg(\calP)+1$ for the amplitudes of $\calP(h)$ for any polynomial
$\calP$ (Proposition~\ref{rem:direct_sum}).  These results demonstrate
an exponential improvement over the connected-component method for
Hamiltonians built from mutually anticommuting generators
(Proposition~\ref{prop:anticommuting}).

A central theme of this work is that pilot state preparation is only one
component of the HDQI pipeline: efficient decoding of the associated
Hamiltonian code is equally essential, and the two must be considered in
conjunction.  The predecessor-uniform condition solves the pilot state
problem, but its tension with generator locality
(Remark~\ref{rem:locality}) constrains the class of applicable
Hamiltonians.

To illustrate the subtlety of combining pilot state preparation with
decoding, consider the perturbed toric code
$H = -\sum_v A_v - \sum_p B_p + \epsilon\sum_j Z_j$ on an
$L \times L$ torus.  If the plaquette operators $B_p$ and the
single-qubit $Z_j$ are all included as generators, then $B_p$ is the
product of four generators $Z_{j_1}Z_{j_2}Z_{j_3}Z_{j_4}$, yielding a
weight-5 relation $B_p \cdot Z_{j_1}\cdot Z_{j_2}\cdot Z_{j_3}\cdot
Z_{j_4} = I$ in the generator group.  Consequently, the Hamiltonian
code has distance at most~$5$, independent of the lattice size~$L$.
Even if one could efficiently prepare the pilot state for this
Hamiltonian, the small code distance would preclude reliable decoding.
Dropping $B_p$ from the generator set restores the code distance but
makes the Hamiltonian non-linear in the remaining generators, taking it
outside the current framework.  This example illustrates why identifying
Hamiltonians that simultaneously admit efficient pilot state preparation
\emph{and} sufficient code distance remains a challenging open problem.

On the combinatorial side, it would be interesting to explore the
computational complexity of evaluating or approximating the twisted
multinomial coefficient for general twisting matrices~$\Omega$, and to
understand more precisely the role that predecessor-uniformity plays as a
boundary between tractable and potentially intractable instances.

We hope that the combinatorial perspective developed here --- in
particular, the identification of the twisted multinomial as an object
of independent interest --- will shed further light on what makes pilot
state preparation tractable or hard in general, and will guide the search
for Hamiltonians where the full HDQI pipeline can be realized efficiently.

\medskip
\textbf{Acknowledgements.}
The author thanks Yihui Quek for bringing~\cite{Bu2026} to our attention
and for stimulating discussions about Hamiltonian DQI,
and Christophe Piveteau for pointing out the tension between
predecessor-uniformity and generator locality (Remark~\ref{rem:locality}).

\appendix
\section{Induction Proof of the Twisted Multinomial Theorem}
\label{app:induction}

We give an alternative proof of Theorem~\ref{thm:twisted_multinomial}
by induction on $k$.  This proof is self-contained: it verifies the
identity directly from the shuffle definition
(Definition~\ref{def:twisted_multinomial}) and the algebra relations,
without appealing to confluence.

\begin{proof}[Proof of Theorem~\ref{thm:twisted_multinomial} by induction]
  \textit{Base case ($k=1$).}  We have $h = \sum_{j=1}^m c_j z_j$.  The
  only compositions of $1$ into $m$ parts have $k_j = 1$ for a single
  index~$j$ and $k_i = 0$ for $i \ne j$.  Each such shuffle
  consists of a single letter, so $\Sh(0,\ldots,0,1,0,\ldots,0)$
  contains exactly one element with no inversions.  Hence
  $\binom{1}{0,\ldots,0,1,0,\ldots,0}_\Omega = 1$, and the right-hand
  side of~\eqref{eq:twisted_expansion} equals $\sum_j c_j z_j = h$.

  \textit{Inductive step.}  Assume \eqref{eq:twisted_expansion} holds
  for $k-1$.  Then
  \begin{align}
    h^k = h \cdot h^{k-1}
    &= \biggl(\sum_{i=1}^m c_i z_i\biggr)
    \sum_{\substack{k'_1+\cdots+k'_m=k-1}}
    \binom{k{-}1}{k'_1,\dots,k'_m}_\Omega
    \prod_{j=1}^m c_j^{k'_j}\,z_1^{k'_1}\cdots z_m^{k'_m}.
    \label{eq:induction_step}
  \end{align}
  Consider a single product $z_i \cdot z_1^{k'_1}\cdots z_m^{k'_m}$
  from this sum.  To bring $z_i$ to its canonical position, it must move
  past $z_1^{k'_1}\cdots z_{i-1}^{k'_{i-1}}$ to the right.  Each copy
  of $z_j$ with $j < i$ produces a factor $\omega_{j,i}$ (from
  $z_i z_j = \omega_{j,i}\,z_j z_i$), so the total phase is
  $\prod_{j<i}\omega_{j,i}^{k'_j}$ and the product becomes
  \begin{align}
    z_i \cdot z_1^{k'_1}\cdots z_m^{k'_m}
    =
    \biggl(\prod_{j<i}\omega_{j,i}^{k'_j}\biggr)
    z_1^{k'_1}\cdots z_{i-1}^{k'_{i-1}}\,z_i^{k'_i+1}\,
    z_{i+1}^{k'_{i+1}}\cdots z_m^{k'_m}.
  \end{align}
  Setting $k_j = k'_j$ for $j \ne i$ and $k_i = k'_i + 1$ and collecting
  all contributions to a given composition $(k_1,\ldots,k_m)$ of $k$, the
  coefficient of $\prod_j c_j^{k_j}\,z_1^{k_1}\cdots z_m^{k_m}$ becomes
  \begin{align}
    \sum_{\substack{i=1\\k_i \ge 1}}^m
    \binom{k{-}1}{k_1,\ldots,k_i{-}1,\ldots,k_m}_\Omega
    \prod_{j<i}\omega_{j,i}^{k_j}.
    \label{eq:recurrence}
  \end{align}
  It remains to verify that~\eqref{eq:recurrence} equals
  $\binom{k}{k_1,\ldots,k_m}_\Omega$.  Partition the shuffles in
  $\Sh(k_1,\ldots,k_m)$ according to the letter $i = \sigma(1)$
  occupying the first position.  Removing position~$1$ yields a shuffle
  of $(k_1,\ldots,k_i{-}1,\ldots,k_m)$ with the same internal
  inversions and hence the same internal phase, which contributes
  $\binom{k{-}1}{k_1,\ldots,k_i{-}1,\ldots,k_m}_\Omega$.  The inversions
  involving position~$1$ are the pairs $(1,s)$ with $\sigma(s) < i$;
  in the remaining $k-1$ positions, letter~$j$ with $j < i$ appears
  $k_j$ times, each contributing a weight $\omega_{j,i}$.  The total
  phase from these inversions is $\prod_{j<i}\omega_{j,i}^{k_j}$.
  Summing over all choices of~$i$ recovers~\eqref{eq:recurrence},
  confirming the identity.
\end{proof}


\end{document}